\documentclass[preprint,11pt]{elsarticle}

\usepackage{amssymb}

\usepackage{graphicx}
\usepackage{color}
\usepackage{fancyhdr,fancybox}
\usepackage{varioref,float,amssymb,amsmath}
\usepackage{float,amsfonts,amsthm}
\usepackage{changebar,longtable}
\usepackage{makeidx}
\usepackage{booktabs}
\usepackage{hyperref}
\hypersetup{colorlinks,linkcolor=black,citecolor=blue,urlcolor=blue}
\usepackage{mathdots}
\usepackage{verbatim}
\usepackage{booktabs}
\usepackage{pdflscape}
\usepackage[utf8]{inputenc}
\usepackage{colortbl} 
\usepackage{xcolor}   
\usepackage{soul}
\usepackage[utf8]{inputenc} 
\usepackage[T1]{fontenc}
\usepackage{multirow} 
\usepackage{colortbl} 

\newcommand{\bra}[1]{{\left\langle{#1}\right\vert}}
\newcommand{\ket}[1]{{\left\vert{#1}\right\rangle}}

\newtheorem{prop}{Proposition}

\newtheorem{lemma}[prop]{Lemma}

\journal{ArXiv}

\begin{document}

\begin{frontmatter}

\title{Single-Qudit Quantum Neural Networks for Multiclass Classification}

\author[inst1,inst2]{Leandro C. Souza}
\author[inst1,inst3]{Renato Portugal}

\affiliation[inst1]{organization={National Laboratory of Scientific Computing, LNCC},
            addressline={Av. Getulio Vargas, 333}, 
            city={Petr\'{o}polis},
            postcode={25651-075}, 
            state={RJ},
            country={Brazil}}

\affiliation[inst2]{organization={Universidade Federal da Paraíba, UFPB},
            addressline={Rua dos Escoteiros, s/n}, 
            city={Jo\~{a}o Pessoa},
            postcode={58051-900}, 
            state={PB},
            country={Brazil}}

\affiliation[inst3]{organization={Universidade Católica de Petrópolis, UCP},
            addressline={Rua Bar\~{a}o do Amazonas, 124}, 
            city={Petr\'{o}polis},
            postcode={25685-100}, 
            state={RJ},
            country={Brazil}}

\begin{abstract} 
This paper proposes a single-qudit quantum neural network for multiclass classification, by using the enhanced representational capacity of high-dimensional qudit states. Our design employs an $d$-dimensional unitary operator, where $d$ corresponds to the number of classes, constructed using the Cayley transform of a skew-symmetric matrix, to efficiently encode and process class information. This architecture enables a direct mapping between class labels and quantum measurement outcomes, reducing circuit depth and computational overhead. To optimize network parameters, we introduce a hybrid training approach that combines an extended activation function---derived from a truncated multivariable Taylor series expansion---with support vector machine optimization for weight determination. We evaluate our model on the MNIST and EMNIST datasets, demonstrating competitive accuracy while maintaining a compact single-qudit quantum circuit. Our findings highlight the potential of qudit-based QNNs as scalable alternatives to classical deep learning models, particularly for multiclass classification. However, practical implementation remains constrained by current quantum hardware limitations. This research advances quantum machine learning by demonstrating the feasibility of higher-dimensional quantum systems for efficient learning tasks.

\

\end{abstract}

\end{frontmatter}



\section{Introduction}\label{sec-intro}

Quantum Machine Learning (QML) merges quantum computing and machine learning to solve computational problems more efficiently than classical approaches~\cite{Biamonte2017,Wang2024,PERALGARCIA2024}. Using the principles of quantum mechanics, QML models can explore vast solution spaces simultaneously, offering the potential for speedups in tasks like optimization, pattern recognition, and data classification~\cite{Huang2021,Liu2024}. QML algorithms often rely on quantum circuits that map data into high-dimensional quantum states, enabling the discovery of patterns and relationships that may be intractable for classical systems. These capabilities are particularly useful for large-scale applications, including image recognition, natural language processing~\cite{Minaee2024}, and drug discovery~\cite{Dara2022}. Although still in its early stages, QML has shown promise in enhancing classical machine learning by improving data representation, accelerating computation, and introducing novel feature extraction techniques~\cite{Benedetti2019a,Huang2021a}.

Quantum Neural Networks (QNNs) are a class of QML algorithms inspired by classical neural networks, designed to exploit the computational advantages of quantum mechanics for learning tasks. QNNs use quantum states and operations to encode, process, and transform data in ways that classical systems cannot efficiently replicate~\cite{Schuld2014}. In classification tasks, QNNs encode input data into quantum states and apply parameterized quantum circuits to perform transformations, effectively learning decision boundaries in high-dimensional quantum space~\cite{Bokhan2022,Mordacci2024,Dhara2024,Ding2024a,Pillay2024}. The probabilistic nature of quantum measurements is used to classify data points into specific categories. This approach allows QNNs to tackle complex classification problems while potentially requiring fewer resources than classical neural networks due to their ability to process high-dimensional data more efficiently. QNNs are particularly promising for tasks involving large and complex datasets, offering improved scalability and performance in domains such as image recognition, anomaly detection, and natural language processing~\cite{Mari2020,Ngairangbam2022,Gujju2024}.

Parameterized Quantum Circuits (PQCs) are a foundational component of quantum machine learning, serving as flexible, trainable quantum models for various tasks~\cite{Benedetti2019,Carvalho2024}. A PQC consists of a series of quantum gates whose parameters can be adjusted to optimize the circuit's behavior, allowing it to learn patterns in data. PQCs are versatile and can approximate complex functions, making them a natural choice for tasks like classification, regression, and generative modeling~\cite{Ding2024}. In the context of QNNs, PQCs play a central role, as QNNs are often implemented using parameterized circuits where the gates act as the neurons in the network. Analogous to classical neural networks, the parameters of these gates depend on trainable weights that are updated during training to minimize a loss function. The flexibility of PQCs enables QNNs to use quantum properties to explore high-dimensional spaces and capture intricate data relationships. This synergy between PQCs and QNNs highlights their potential to provide quantum advantages in machine learning~\cite{Sim2019,Du2020,Schuld2021}.

A qudit is a fundamental unit of quantum information that generalizes the concept of a qubit~\cite{Wang_2020}. While a qubit operates in a two-dimensional quantum state space (spanned by states $\ket{0}$ and $\ket{1}$), a qudit extends this to an $d$-dimensional state space, represented by states $\ket{0}, \ket{1}, \dots, \ket{d-1}$. This higher-dimensional nature allows qudits to encode and process more information per unit, making them inherently more efficient for problems requiring richer data representation. In quantum machine learning, qudits provide a natural advantage for tasks like multiclass classification, where each class can correspond to one of the $d$ dimensions, simplifying the encoding and processing of multiclass datasets~\cite{Wach2023,RocaJerat_2024,Mandilara2024}. In contrast, qubits would require multiple entangled states or ancillary qubits to achieve the same functionality, potentially increasing circuit complexity and resource requirements. Recent experimental advancements in qudit-based quantum processors, including implementations with trapped ions~\cite{Nilolaeva2024} and photonic platforms~\cite{Chi2022}, further support the practical viability of qudit systems. By leveraging these higher-dimensional quantum units, quantum neural networks can offer more direct and scalable solutions for classifying data with multiple categories, enhancing the practical utility of quantum machine learning in real-world applications.

In this work, we propose a quantum neural network architecture based on neurons modeled with a single qudit, by exploiting the advantages of higher-dimensional Hilbert spaces and generalizing the concept of single-qubit models~\cite{Ghobadi2019,Karimi2023,Salinas2020,Salinas2021,Yu2022,McCaldin2024,McFarthing2024,Cuellar2024,Souza2024}. Unlike single-qubit models, which require additional resources for multi-dimensional inputs, the qudit-based neural network inherently supports them due to its $d$-dimensional state space. Furthermore, this framework provides a powerful and versatile architecture for multiclass classification tasks, efficiently encoding data into a single qudit state and applying parameterized unitary operations for transformation. By focusing on single-qudit neurons, this approach reduces the complexity of quantum circuits while maintaining the expressive power needed for high-dimensional data representation. The use of probabilistic activation functions through quantum measurements enhances the model’s suitability for multiclass classification problems. Additionally, employing qudits for multiclass classification minimizes hardware requirements and simplifies implementation, making this method particularly relevant for current and near-term quantum hardware.

The theoretical foundation of our model lies in the parameterization of unitary operators for single qudits, enabling flexible transformations of input data within the $d$-dimensional Hilbert space. For training, we employ an activation function based on a truncated multivariable Taylor series. To determine the weights, we use a regularized support vector machine method, which optimally separates classes while preventing overfitting. This approach effectively approximates complex relationships within the data by using a structured optimization framework that balances margin maximization and classification accuracy. Once the network has been trained, it runs on a quantum computer, but the qudit parameters are obtained using classical computations. Specifically, after training, the optimal weights are known, and the ansatz of the activation function is used to compute the parameter values that are passed to the qudit for execution. This hybrid approach ensures that the quantum model uses both classical optimization techniques and quantum state transformations efficiently.

We evaluated our model on the EMNIST dataset, achieving strong results in multiclass classification. The qudit-based neural network achieved high accuracy, particularly in digit classification, where it reached 98.88\% on the EMNIST Digits dataset and 98.20\% on the EMNIST MNIST dataset. Even for the more challenging EMNIST Letters dataset, the model maintained competitive performance, achieving 89.90\% accuracy. The results also indicate that reducing input dimensionality using Principal Component Analysis (PCA), specifically retaining 20–30 principal components, provides an optimal balance between accuracy and computational efficiency. These results demonstrate the scalability and adaptability of single-qudit neurons, emphasizing their practical potential in quantum machine learning.

The remainder of this paper is organized as follows. In Section~\ref{sec-neuron}, we introduce the qudit-based neuron model, detailing its mathematical formulation and quantum properties. Section~\ref{sec-qudit-NN} extends this framework to define the single-qudit quantum neural network, outlining its architecture and computational advantages. The training methodology for the proposed network is described in Section~\ref{sec-training}, where we introduce the activation function design, weight optimization, and class determination strategies. In Section~\ref{sec-emnist}, we present experimental results on the classification of the EMNIST dataset, analyzing the performance of the qudit-based neural network across different subsets. Section~\ref{sec-orig-mnist} further evaluates the model's capability by applying it to the original MNIST dataset. Finally, Section~\ref{sec-conc} provides concluding remarks and discusses potential directions for future research.

\section{Qudit-Based Neuron Model}\label{sec-neuron}

A qudit is a quantum system with a $d$-dimensional Hilbert space, where $d>2$, generalizing the concept of a qubit to accommodate more than two basis states. The state of a qudit can be expressed as a superposition of $d$ orthonormal basis states as 
$$\ket{\psi} = \sum_{i=0}^{d-1} c_i \ket{i},$$
where $c_i$ are complex coefficients satisfying $\sum_{i=0}^{d-1} |c_i|^2 = 1$. Operations on qudits are performed using $d \times d$ unitary matrices, which generalize qubit gates such as the Pauli and Hadamard gates. For example, a generalized $d$-dimensional Pauli-X gate, often called the $X_d$ gate, cyclically permutes the basis states: $X_d \ket{i} = \ket{(i+1) \bmod d}$. Similarly, a generalized $d$-dimensional Hadamard gate can be defined to create equal superpositions across all $d$ states.

Qudit-based quantum circuits require fewer logical units to encode the same amount of information as qubit-based circuits, reducing the overall circuit depth and complexity. Additionally, the ability to operate in higher-dimensional state spaces provides a natural framework for representing more complex data structures and enhances the expressiveness of quantum algorithms. These properties make qudits particularly advantageous in applications like error correction, where higher-dimensional codes can improve fault tolerance, and in quantum simulation, where they more efficiently model physical systems with intrinsic $d$-level behaviors.

A single qudit does not exhibit entanglement because entanglement is a property associated with correlations between multiple subsystems in a composite quantum system. However, the $d$-dimensional Hilbert space of a qudit can have analogous properties that are often related to quantum correlations. These include the capacity for complex superpositions and rich information structures, which can be exploited in quantum technologies. For instance, measures such as mutual information and entropic inequalities, which is often used to describe composite systems, can be adapted to characterize the informational and quantum properties of a single qudit, allowing it to serve as a resource in advanced quantum applications~\cite{Manko2014}.

A quantum neuron model based on a single qudit with $d$ levels is particularly well-suited for supervised learning tasks, especially those involving multiclass classification. Consider a dataset of the form $\left\{ \left({x}_{1},y_{1}\right),\left({x}_{2},y_{2}\right),\cdots,\left({x}_{n},y_{n}\right)\right\} \subset\mathbb{R}^{p}\times [-1,1],$ where each input ${x}_i \in \mathbb{R}^p$ is a feature vector, and each corresponding output $y_i \in [-1,1]$ represents the target value. The quantum neuron uses quantum computational principles to predict the output $y$ for a previously unseen input ${x} \in \mathbb{R}^p$. The input to the neuron is encoded into an arbitrary qudit state $\ket{\psi}$ in a $d$-dimensional Hilbert space. The neuron is then modeled by a parameterized $d \times d$ unitary matrix $U\big(\vec\theta\,\big)$, where $\vec\theta=(\theta_1,...,\theta_{d-1})$ is a vector of parameters that are adjusted during training. The quantum operation applies $U\big(\vec\theta\,\big)$ to the input state, resulting in the transformation $U\big(\vec\theta\,\big)\ket{\psi}$. A measurement in the computational basis is performed, yielding one of $d$ possible outcomes, which correspond to the predicted class labels. This structure makes the model naturally compatible with multiclass classification tasks, efficiently using the high-dimensional state space of the qudit.

Fig.~\ref{fig-percepton} illustrates the single-qudit neuron model, which is inherently well-suited for multiclass classification tasks. Assuming that there are $d$ classes, each class is mapped to one of the $d$ possible measurement outcomes, with the probability of outcome $\ell$ given by $|\bra{\ell}U\big(\vec\theta\,\big)\ket{\psi}|^2$ for $0\le \ell <d$. This probabilistic measurement process captures the model's confidence in its predictions, while the parameterized unitary matrix $U\big(\vec\theta\,\big)$ enables the neuron to learn complex decision boundaries in the high-dimensional quantum state space. By using the enhanced representational capacity of qudits, the qudit-based quantum neuron achieves efficient encoding and accurate classification of multiclass data.

\begin{figure}[H]
\centering\includegraphics[scale=0.19]{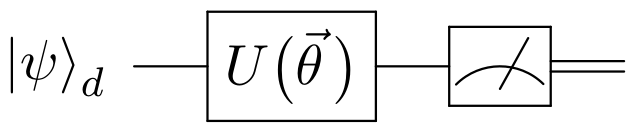}
\caption{Illustration of a qudit-based quantum neuron. The input can be an arbitrary $d$-dimensional quantum state $\ket{\psi}_d$, which undergoes a transformation via the parameterized unitary gate $U\big(\vec\theta\,\big)$. A measurement in the computational basis produces one of $d$ possible outcomes, corresponding to the neuron's prediction.}\label{fig-percepton} 
\end{figure}

Any $d$-dimensional unitary operator $U$ can represent a quantum neuron. It is known that $U$ has $d^2$ independent real parameters, as the unitary constraint imposes $d^2$ real constraints on the $2d^2$ real parameters of an arbitrary $d$-dimensional complex matrix. The unitary constraint ensures that each row (or column) of the matrix is orthonormal. However, there is no known algebraic parametrization for an arbitrary $d$-dimensional unitary operator. Alternatively, the Cayley transform~\cite{Bha97} provides a mapping between Hermitian matrices and unitary matrices, making it a practical method for constructing unitary matrices from Hermitian matrices, which are easier to generate. If $H$ is a Hermitian matrix, then the Cayley transform is given by
\begin{equation}\label{eq-U_H}
    U = (H-\textrm{i} I)(H+\textrm{i} I)^{-1}
\end{equation}  
where $U$ is guaranteed to be unitary. Furthermore, if $H = \textrm{i} A$, where $A$ is a real skew-symmetric matrix, the resulting $U$ is an orthogonal matrix.

In this work, we use a specific unitary operator $U$, which is an orthogonal matrix parameterized by $d-1$ real parameters $\theta_k$, where $k$ ranges from 1 to $d-1$. Since $U$ is derived from a Hermitian matrix $H = \textrm{i} A$, we now describe the entries of the skew-symmetric matrix $A$. The non-zero entries of $A$ are confined to its first row and first column. As $A$ is skew-symmetric, the diagonal entry $A_{11}$ is zero. To fully define $A$, it suffices to describe its first row. The entries of the first row are determined using auxiliary expressions defined as
\begin{align}\label{eq-c-and-s}
\begin{split}
c_\ell &=\cos\theta_{d-\ell}, \\
s_\ell &= \prod_{k=1}^{d-\ell} \sin \theta_k,
\end{split}
\end{align}
for $\ell$ ranging from 1 to $d$, with the convention that $c_0=s_d = 1$. Using these auxiliary expressions, the entries of the first row are given by
\begin{equation}\label{eq-A}
    A_{1\ell}=\frac{s_\ell \, c_{\ell-1}}{s_1-1},
\end{equation}
for $\ell$ ranging from 2 to $d$.

In our experiments, we always use $\ket{0}$, the first vector of the computational basis, as the input to the neuron because the encoding comes through the parameters. Then, the output of the neuron before measurement is $U\ket{0}$, which corresponds to the first column of the unitary matrix $U$. In~\ref{app:A}, we show that
\begin{equation}\label{eq-psi}
U\ket{0} \,=\, \sum_{\ell=0}^{d-1} s_{\ell+1} c_{\ell}\ket{\ell}.
\end{equation}
After measurement, the probabilities of the outcomes correspond to the squared coefficients of the computational basis vectors. The probability of classifying the input into the first class is $s_1^2$, the probability for the second class is $s_2^2\cos^2\theta_{d-1}$, for the third class $s_3^2\cos^2\theta_{d-2}$, and so on.

For practical implementations, it is important to consider the current limitations of quantum hardware, which predominantly support qubit-based systems. While this work focuses on the theoretical advantages and applications of qudit-based models,~\ref{app:B} shows how to construct a qubit-based circuit that can simulate the state transformations and measurement outcomes of the single-qudit neuron discussed above. This approach ensures compatibility with existing quantum platforms while preserving the key features of the qudit model.

\section{Single-Qudit Quantum Neural Network}\label{sec-qudit-NN}

The simplest single-qudit quantum neural network is constructed by streamlining qudit-based neurons, each parameterized by $d-1$ parameters. In this configuration, the network's input is an arbitrary $d$-dimensional quantum state $\ket{\psi}_d$, which is sequentially transformed by a series of parameterized unitary operations, $U_\ell\big(\vec{\theta}^{(\ell)}\big)$, where $\vec{\theta}^{(\ell)} = (\theta^{(\ell)}_1, \dots, \theta^{(\ell)}_{d-1})$ represents the set of parameters for the $\ell$-th layer, with $\ell = 1, \dots, L$. Due to the streamlined design, the number of layers is equal to the number of neurons. The network's output is obtained through a measurement in the computational basis, yielding one of $d$ possible outcomes. Figure~\ref{fig-QuditQNN} illustrates this architecture.

\begin{figure}[H]
\centering\includegraphics[scale=0.19]{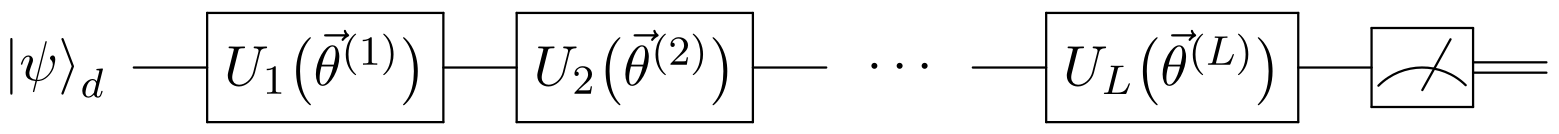}
\caption{Depiction of a single-qudit quantum neural network with $L$ layers. The input is a $d$-dimensional quantum state $\ket{\psi}_d$, which is sequentially transformed by parameterized unitary gates $U_\ell\big(\vec{\theta}^{(\ell)}\big)$ for $\ell = 1, \cdots, L$. A measurement in the computational basis at the end of the network produces one of $d$ possible outcomes, corresponding to the network's prediction.}\label{fig-QuditQNN} 
\end{figure}

Each unitary gate $U_\ell\big(\vec{\theta}^{(\ell)}\big)$ is a parameterized $d \times d$ matrix that performs arbitrary transformations in the $d$-dimensional Hilbert space. These gates can be implemented using generalizations of rotation gates or controlled operations for qudits. The trainable parameters $\vec{\theta}^{(\ell)}$ allow the network to adapt to the underlying data distribution, learning complex patterns for classification and regression tasks. By extending the parameterization to include functions of the input data, the network gains additional flexibility, enabling it to effectively model nonlinear relationships and enhance its expressive power in quantum machine learning.

The structure of the single-qudit quantum neural network is defined by the product of the unitary matrices across its $L$ layers
\[
U_\text{network} = U_L\big(\vec{\theta}^{(L)}\big) \cdots \,U_2\big(\vec{\theta}^{(2)}\big) \,U_1\big(\vec{\theta}^{(1)}\big),
\]
which represents the cumulative transformation applied to the input state $\ket{\psi}_d=\ket{0}$. The depth $L$ of the network determines its expressive power, allowing the model to approximate increasingly complex functions as $L$ increases. The matrix $U_\text{network}$ is an effective $d$-dimensional unitary transformation, which can be parameterized with $d-1$ parameters by combining all ansatze used in $\vec{\theta}^{(\ell)}$, for $\ell$ ranging from 1 to $L$.

The model’s inherent capability for multiclass classification arises from the qudit's dimensionality $d$, where each measurement outcome naturally corresponds to a class. This eliminates the need for additional qubits or post-processing steps required in qubit-based networks to encode multiclass labels. Moreover, the higher-dimensional structure of qudits reduces overall network complexity by allowing each neuron to process more information.

\section{Training Method}\label{sec-training}

Training a single-qudit quantum neural network involves systematically improving its performance on specific tasks by identifying patterns in a dataset. This process is guided by optimizing the network’s parameters, which control the behavior of the unitary gates within the model. The primary objective is to minimize a loss function that quantifies the discrepancy between the network’s predictions and the true outcomes in the training data. To achieve this, we employ a regularized Support Vector Machine (SVM) optimization, which determines the optimal weights while preventing overfitting. Unlike purely gradient-based learning methods, this approach provides a structured optimization framework that ensures efficient training while leveraging the high-dimensional state space of qudits. The ultimate goal is to train a network that generalizes well to unseen data, capturing the essential structure and relationships within the dataset.

\subsection{Activation Function Design}

The network circuit is parameterized by a set of trainable parameters, denoted as $\vec\theta = (\theta_1, \dots, \theta_{d-1})$ for each neuron. Let $\theta$ represent a generic parameter or component of $\vec\theta$ in the network. To implement the learning algorithm, each $\theta$ is considered as a function mapping from $\mathbb{R}^p$ to $\mathbb{R}$, representing a trainable weighted function derived from the dataset inputs. This approach, known as data re-uploading, has been explored in works such as Refs.~\cite{Salinas2020, Souza2024}. 

In Ref.~\cite{Salinas2020}, the authors propose an ansatz of the form $U(\vec\phi + \vec w \circ \vec x)$ for re-uploading classical information $\vec x$, where $\vec w$ is a tunable weight vector and $\vec\phi$ is a tunable parameter acting as a bias. Id this ansatz, the operation between $\vec w$ and the input data $\vec x$ is the Hadamard (elementwise) product. The neural network used in this reference is based on a single qubit, meaning that an arbitrary unitary operator in $SU(2)$ has at most three degrees of freedom, modulo a global phase. As a result, the input dimension is limited to a maximum of three. To address this limitation, Ref.~\cite{Salinas2020} proposes partitioning the $p$-dimensional vector $\vec{x}$ into approximately $p/3$ smaller vectors by increasing the number of layers in the network. Using qudits instead of qubits reduces the number of required partitions because qudits can encode more information per unit.

In contrast, Ref.\cite{Souza2024} adopts a different data re-uploading strategy. Instead of using the Hadamard product, their ansatz employs the dot product, expressed in its simplest form as $\cos\theta(\vec{x}) = \tanh(w_0 + \vec{w} \cdot \vec{x})$, where $\theta$ is an entry of $\vec{\theta}$. Note that the dot product does not restrict the dimensionality of the input vector. This ansatz closely resembles the one universally employed in the classical perceptron~\cite{rosenblatt1958}, suggesting that each quantum neuron is associated with a classical perceptron. The perceptron performs calculations on a classical device before delivering the parameter value to the quantum neuron.

Ref.~\cite{Souza2024} further generalized this ansatz for a single-qubit quantum network with $L$ layers by introducing the activation function 
\begin{equation}\label{eq-theta-}
\cos \theta(\vec x) = \tanh\left(w_0 + \vec w_1 \cdot \vec x +  \vec w_2 \cdot (\vec x\circ \vec x) + \cdots + \vec w_L \cdot \vec x^{\,\circ L}\right),
\end{equation}
where $\vec x^{\,\circ L}$ represents the Hadamard product applied $L$ times. The reasoning behind this ansatz is based on two key notions: (i) a sequential arrangement of quantum neurons is mathematically equivalent to an effective single neuron, and (ii) incorporating powers of the input data entries is inspired by a truncated Taylor series expansion. Eq.~\eqref{eq-theta-} was successfully employed for binary classification tasks across various synthetic and real datasets in Ref.~\cite{Souza2024}.

In this work, we focus on addressing multiclass classification by using the multi-output capabilities of qudit measurements. We adopt a strategy similar to the one used in Ref.~\cite{Souza2024}, parameterizing the effective $U_\text{network}$ with only $d-1$ parameters $\theta_1, \dots, \theta_{d-1}$, as described in Eq.\eqref{eq-U_H} with $H = \textrm{i} A$, where the entries of $A$ are defined by Eq.~\eqref{eq-A}. The parameters depend on the input data in a nonlinear way, with the degree of non-linearity bounded by $L$, which represents the number of neurons. This approach avoids directly uploading classical information to the quantum neurons. Instead, a nonlinear classical ``perceptron'' performs numerical computations on the input data and provides the optimal parameter values to the effective quantum neuron.

Unfortunately, Eq.~\eqref{eq-theta-} did not yield satisfactory results in our experiments with EMNIST and MNIST. To address this limitation, we propose an extended ansatz for each parameter $\theta(x_i)$, where $x_i$ is a feature vector, defined by a new degree-$L$ activation function as 
\begin{equation}\label{eq-ansatz}
\sin\theta(x_i) = \sigma\left( w_0 + \sum_{j=1}^p w_{j} x_{ij} + \sum_{j=1}^p \sum_{k=j}^p w_{j k} x_{ij}x_{ik} + \sum_{j=1}^p \sum_{k=j}^p\sum_{\ell=k}^p w_{j k \ell} x_{ij}x_{ik}x_{i\ell}+ \cdots\right),
\end{equation}
such that the last term inside the sigmoid function~$\sigma$ includes $L$ summations, with the coefficients represented as trainable weights indexed by $L$ labels. The remaining weights are indexed with a number of labels varying from 1 to $L-1$. Using basic combinatorial principles, the total number of trainable weights in Eq.~\eqref{eq-ansatz} is given by $\binom{L+p}{p}$.

This ansatz extends the one in Eq.~\eqref{eq-theta-} by incorporating different input data entries, ensuring that the maximum degree corresponds to the number of layers. The hyperparameter $L$, which represents the number of layers or neurons, serves a similar role as in Eq.~\eqref{eq-theta-}. The formulation is inspired by a truncated multivariable Taylor series expansion. By capturing correlations among input features more effectively, this approach enhances the model performance.

\subsection{Weight Determination}

SVM provides a robust approach for weight optimization by determining an optimal hyperplane that maximizes the margin between different classes. For a given dataset $\{(x_i, y_i)\}_{i=1}^{n}$, where $x_i \in \mathbb{R}^p$ represents the input feature vectors and $y_i \in \{-1,1\}$ are the corresponding class labels, the decision function is given by 
\begin{equation}\label{eq-y-i}
y_i = w \cdot x_i + b,
\end{equation}
where $w \in \mathbb{R}^p$ is the weight vector and $b \in \mathbb{R}$ is the bias term. The optimal weights and bias are obtained by solving the regularized convex optimization problem  
\begin{equation}\label{eq-SVM}
\min_{w, b} \frac{1}{2} \|w\|^2 + C \sum_{i=1}^{n} \xi_i
\end{equation}
subject to the constraints  
\begin{equation}
y_i\, (w\cdot x_i + b) \geq 1 - \xi_i, \quad \xi_i \geq 0, \quad \forall i \in \{1, \dots, n\}.
\end{equation}
Here, $\xi_i$ are slack variables that allow for misclassification, and $C > 0$ is a hyperparameter that controls the trade-off between maximizing the margin and minimizing classification errors. This formulation, known as the soft-margin SVM, ensures robustness against outliers while maintaining good generalization performance. In this work, we adopt the regularized SVM method to determine the optimal weights for the qudit-based quantum neural network, replacing the linear least squares approach used in other references such as~\cite{Souza2024}.

\subsection{Class Determination}

The output of the quantum network is a state vector whose entries are given by $s_{j+1} \cos \theta_{d-j}$ for $0 \leq j < d$, as described in Eqs.~\eqref{eq-c-and-s} and~\eqref{eq-psi}. Consequently, the probability of measuring the system in ``class $j$'' is determined by  
\begin{equation}
p_j = \left( \prod_{k=1}^{d-j-1} \sin^2 \theta_k \right) \cos^2 \theta_{d-j},
\end{equation}
where, by convention, $p_{d-1} = \cos^2 \theta_1$.  

The classification process follows a sequential training approach, starting with the last class, $d-1$, whose associated probability is $p_{d-1}$. When $\theta_1$ is close to zero, $p_{d-1}$ approaches 1, meaning that the probabilities of all other classes are negligible. Since $\theta_1(x_i)$ follows Eq.~\eqref{eq-ansatz}, the weights are determined using Eq.~\eqref{eq-SVM}, where the class labels are converted into two distinct values: $-1$ if the label is $d-1$, and $+1$ otherwise. We then re-scale the computed weights by mapping $-1$ to $-\infty$ and $+1$ to $+\infty$, ensuring that the output of the activation function $\sigma$ is close to either 0 or 1. This transformation effectively forces large negative elements for samples labeled $d-1$, making $\sin^2 \theta_1$ close to zero and increasing $\cos^2 \theta_1$, which maximizes the probability of the network returning class $d-1$. In our experiments, we confirmed that scaling factors such as 100 were sufficiently large for this effect to hold.

After determining $\theta_1$, we proceed to $\theta_2$ using a similar approach. Specifically, setting $\theta_2$ close to zero ensures that the probability $p_{d-2}$ approaches 1, while the probabilities of all other remaining classes become negligible. At this stage, data points associated with class $d-1$ are removed from training, as they are no longer relevant, and $\theta_1$ is set to $\pi/2$. When applying Eq.~\eqref{eq-SVM} again, we exclude all data inputs corresponding to label $d-1$ and continue the training process to determine $\theta_2$. This sequential classification continues until all parameters $\theta_j$ are computed, allowing the model to assign the correct class to each input.

While the training process above assumes a standard correspondence between class labels and measurement results, this mapping can be further optimized. Since the entries of the output vector~\eqref{eq-psi} are not symmetrically distributed in terms of sine and cosine factors, choosing an appropriate class association can improve classification accuracy. The number of sine factors is maximal for the first entry and decreases linearly, reaching zero for the last entry. Conversely, each entry contains exactly one cosine factor, except for the first entry, which has none. Given this structure, an optimized assignment of class labels to measurement outcomes can enhance classification accuracy.

Before training, we determine the optimal mapping between class labels and measurement outcomes. The initial assumption that class $d-1$ is associated with the training of $\theta_1$, where a small $\theta_1$ leads to a high probability $p_{d-1}$, is not necessarily optimal. Instead, we systematically evaluate all possible class-to-measurement assignments and select the one that minimizes classification error, measured by the Hinge loss function. The following algorithm determines the optimal mapping between class labels and measurement results.

The process begins by determining the class corresponding to the training of $\theta_1$. For each class $j$ from 0 to $d-1$, we convert all class labels into a binary format: $y_i' = -1$ if the label corresponds to $j$ and $y_i' = +1$ otherwise. The network is then trained with these modified labels, and the associated weights are computed. Next, we calculate the Hinge loss function, given by
\begin{equation}
\ell_j(y') = \sum_{i=1}^{n} \max(0, 1 - y'_i \hat{y}'_i).
\end{equation}
This process is repeated for all classes, and the class $j$ that yields the smallest Hinge loss is selected as the one associated with $\theta_1$. This ensures that when $\theta_1$ is close to zero, $p_j$ is maximized while all other probabilities remain low. The number of iterations in this step is $d$.

After determining $\theta_1$, we proceed with $\theta_2$ by removing all data entries corresponding to the class previously assigned to $\theta_1$. We then set $\theta_1 = \pi/2$, which modifies the output vector so that only $d-1$ entries remain nonzero, with the last entry being zero and the second-to-last entry being $\cos\theta_2$. The goal now is to determine which class $j'$ should be associated with $\theta_2$ so that when $\theta_2$ is close to zero, $p_{j'}$ is maximized. To achieve this, we repeat the same procedure as before, iterating over the remaining $d-1$ classes, converting labels into binary form, training the network, and computing the Hinge loss. The class corresponding to the smallest Hinge loss is assigned to $\theta_2$. The number of iterations in this step is $d-1$.

This process continues sequentially, determining the optimal class assignment for each $\theta_j$ until all classes are assigned. The final label-to-measurement mapping is obtained using an algorithm that runs in $O(d^2)$ time.

\section{Classification of EMNIST Datasets}\label{sec-emnist}

The EMNIST dataset is a significant extension of the widely used MNIST dataset, designed to include handwritten letters and provide a more comprehensive benchmark for machine learning models. Introduced by Cohen et al.~\cite{Cohen2017}, EMNIST was created by extracting and preprocessing a subset of the NIST Special Database 19,\footnote{\url{https://nvlpubs.nist.gov/nistpubs/Legacy/IR/nistir5959.pdf}} a large collection of handwritten characters collected by the United States National Institute of Standards and Technology (NIST). This original NIST dataset contains a diverse set of handwritten digits and letters obtained from different writers, making it a more challenging and representative dataset than MNIST, which was derived from a more homogeneous subset of NIST data. EMNIST preserves the structure of MNIST by retaining the same 28$\times$28 grayscale image format, ensuring compatibility with models and tools developed for the original dataset. However, unlike MNIST, which only includes digits, EMNIST introduces uppercase and lowercase letters, significantly increasing the dataset's complexity and making it a valuable resource for evaluating models in a broader range of classification tasks.

The EMNIST dataset is divided into six distinct subsets, categorized into two main groups based on their structure and sample distribution. The \textit{By\_Class} and \textit{By\_Merge} subsets represent the raw datasets derived directly from the NIST Special Database 19. The \textit{By\_Class} subset includes all handwritten digits and letters, treating uppercase and lowercase letters as separate classes, resulting in 62 classes in total. In contrast, the \textit{By\_Merge} subset also contains digits and letters but merges uppercase and lowercase letters, reducing the number of classes to 47.

The remaining four subsets are designed to ensure an equal number of samples per class, making them more balanced for training and evaluation purposes. The \textit{Balanced} subset includes 47 classes, merging uppercase and lowercase letters and maintaining an even distribution of samples across all classes. The \textit{Digits} subset consists exclusively of the 10 numerical digits (0--9) and closely mirrors the original MNIST dataset in composition. The \textit{Letters} subset focuses solely on handwritten letters, with uppercase and lowercase versions merged, resulting in 26 classes. Finally, the \textit{EMNIST MNIST} subset is identical to the original MNIST dataset, providing full compatibility for benchmarking against traditional machine learning models. These structured subsets allow for targeted exploration of different classification tasks, making EMNIST a highly versatile and challenging dataset for machine learning and pattern recognition research.

To reduce the dimensionality of the input data while preserving essential features, we applied Principal Component Analysis (PCA) to the EMNIST images. Each image, originally represented as a 784-dimensional vector (28$\times$28 pixels), was projected onto a lower-dimensional space using the leading eigenvectors of the covariance matrix. Specifically, we retained the top 10 to 40 principal components, capturing the most significant variance in the dataset. This dimensionality reduction improved training efficiency while minimizing redundant information. For performance evaluation, we employed 10-fold cross-validation ($K=10$), ensuring that the model was trained and tested on multiple data splits for a robust accuracy estimate. The impact of PCA and cross-validation on classification performance is reflected in the results presented in the following sections. No additional image preprocessing beyond PCA was applied, ensuring that the classification results reflect the raw dataset without modifications.

In the following sections, we analyze the performance of the qudit-based neural network on each dataset individually. For each subset, we present classification results, highlighting accuracy, computational efficiency, and the impact of different network configurations. The implementation was carried out using Python and scikit-learn, and all experiments were conducted on a desktop PC with an AMD Ryzen 9 5950X 16-Core 3.4GHz processor and 128GB of RAM.

\subsection{Balanced Dataset}

The Balanced dataset is designed to provide a balanced and streamlined resource for handwritten letters and digits recognition tasks. It is intended to be the most widely applicable dataset among the other datasets in EMNIST~\cite{Cohen2017}. Specifically, the Balanced dataset contains 131,600 samples, evenly distributed across 47 classes: 10 digits (0-9) and 37 letter classes, where certain similar uppercase and lowercase letters (like 'C' and 'c') are merged to reduce confusion and simplify classification. This balanced split, with approximately 2,800 samples per class (2,400 for training and 400 for testing), ensures uniform representation, making it an ideal benchmark for training and evaluating machine learning models on a set of handwritten characters.

\begin{table}[H]
\centering
\begin{tabular}{|c|c|c|c|c|}
\hline
\rowcolor{gray!20} \textbf{Components} & \textbf{Neurons} & \textbf{Weights} & \textbf{Accuracy (\%)} & \textbf{Time (s)} \\
\hline
\multirow{3}{*}{\textbf{10}} & 1 & 11 & 28.22 (0.35) & 37.35 \\
                             & 2 & 66 & 60.29 (0.29) & 232.05 \\
                             & 3 & 286 & 68.75 (0.49) & 1341.98 \\
\hline
\multirow{3}{*}{\textbf{20}} & 1 & 21 & 48.30 (0.23) & 59.93 \\
                             & 2 & 231 & 78.76 (0.21) & 744.94 \\
                             & 3 & 1771 & 81.15 (0.23) & 7591.03 \\
\hline
\multirow{3}{*}{\textbf{30}} & 1 & 31 & 53.07 (0.21) & 84.79 \\
                             & 2 & 496 & 81.26 (0.23) & 1519.69 \\
                             & 3 & 5456 & 81.50 (0.34) & 21021.78 \\
\hline
\end{tabular}
\caption{Results for the Balanced dataset. }
\label{tab:balanced_results}
\end{table}

Table~\ref{tab:balanced_results} presents the classification accuracy and computation time for different numbers of principal components and neurons. The number of neurons is the value of $L$ in Eq.~\eqref{eq-ansatz}. The results indicate that increasing the number of components improves accuracy, with a significant gain from 10 to 20 components, while the improvement from 20 to 30 components is marginal. Similarly, increasing the number of neurons consistently enhances performance, with the largest gain occurring from 1 to 2 neurons, whereas the improvement from 2 to 3 neurons is less pronounced. However, this increased complexity comes at the cost of significantly higher computational time, particularly when using three neurons, where the number of weights quickly grows. The highest accuracy (81.50\%) is obtained with 30 components and 3 neurons, but at a high computational cost (21021.78s), while a similar accuracy (81.26\%) is achieved with 20 components and 2 neurons in significantly less time (1519.69s). This suggests that using 20 components and 2 neurons provides the best trade-off between accuracy and efficiency, as further increasing the model complexity yields only minor improvements in performance while drastically increasing computation time.

\subsection{Digits Dataset}

The Digits dataset is a subset of the EMNIST dataset that exclusively contains handwritten digits (0-9). It is designed as an extension of the original MNIST dataset but is derived from a different source, ensuring greater diversity in handwriting styles. This subset provides a larger and more representative sample of digit images, making it particularly useful for evaluating machine learning models in digit classification tasks. The images retain the standard 28×28 grayscale format, ensuring compatibility with models trained on MNIST. By focusing solely on numerical characters, the EMNIST Digits dataset serves as an effective benchmark for assessing the accuracy and robustness of quantum neural networks in handwritten digit recognition.

\begin{table}[H]
\centering
\begin{tabular}{|c|c|c|c|c|}
\hline
\rowcolor{gray!20}  \textbf{Components} & \textbf{Neurons} & \textbf{Weights} & \textbf{Accuracy (\%)} & \textbf{Time (s)} \\
\hline
\multirow{3}{*}{\textbf{10}} & 1 & 11  & 82.13 (0.66) & 6.87  \\
                             & 2 & 66  & 92.70 (0.16) & 42.86  \\
                             & 3 & 286 & 94.87 (0.16) & 251.93  \\
\hline
\multirow{3}{*}{\textbf{20}} & 1 & 21  & 88.52 (0.18) & 12.08  \\
                             & 2 & 231 & 97.54 (0.09) & 105.99 \\
                             & 3 & 1771 & 98.43 (0.08) & 521.11  \\
\hline
\multirow{3}{*}{\textbf{30}} & 1 & 31  & 90.19 (0.18) & 16.54  \\
                             & 2 & 496 & 98.37 (0.11) & 182.30 \\
                             & 3 & 5456 & 98.88 (0.07) & 1449.55  \\
\hline
\end{tabular}
\caption{Results for Digits dataset}
\label{tab-Digits}
\end{table}

Table~\ref{tab-Digits} presents the classification results for the Digits dataset, evaluating different numbers of principal components and neurons. The results show a consistent improvement in accuracy as the number of components increases, with a notable gain from 10 to 20 components, while the improvement from 20 to 30 components is relatively minor. Similarly, increasing the number of neurons enhances performance, with the most significant jump occurring from 1 to 2 neurons, whereas the gain from 2 to 3 neurons is less pronounced. However, this increase in model complexity also leads to higher computational costs, particularly when using three neurons, where the number of weights and training time grow substantially. The highest accuracy (98.88\%) is achieved with 30 components and 3 neurons, but this configuration requires 1449.55s of training time. A similar accuracy (98.43\%) is obtained with 20 components and 3 neurons in just 521.11s, making it a more computationally efficient choice. These results suggest that using 20 components and 3 neurons provides the best balance between accuracy and training efficiency, as further increasing the number of components yields only marginal gains while significantly increasing computational time.

\subsection{Letters Dataset}

The Letters dataset is a subset of the EMNIST dataset that focuses exclusively on handwritten letters, making it a valuable benchmark for character recognition tasks. It consists of 26 classes corresponding to the English alphabet, where uppercase and lowercase versions of the same letter are merged to reduce ambiguity and improve classification consistency. The dataset retains the standard 28×28 grayscale image format, ensuring compatibility with digit classification datasets while introducing greater complexity due to the variations in letter shapes and styles. This subset provides a challenging test for quantum neural networks, as letter recognition typically involves more intricate patterns than digit classification, requiring models to capture subtle differences between similar characters.

\begin{table}[H]
\centering
\begin{tabular}{|c|c|c|c|c|}
\hline
\rowcolor{gray!20}  \textbf{Components} & \textbf{Neurons} & \textbf{Weights} & \textbf{Accuracy (\%)} & \textbf{Time (s)} \\
\hline
\multirow{3}{*}{\textbf{10}} & 1 & 11  & 41.64 (0.29) & 13.54  \\
                             & 2 & 66  & 69.04 (0.50) & 89.33  \\
                             & 3 & 286 & 77.13 (0.30) & 567.21 \\
\hline
\multirow{3}{*}{\textbf{20}} & 1 & 21  & 54.28 (0.55) & 22.40  \\
                             & 2 & 231 & 85.34 (0.29) & 283.83 \\
                             & 3 & 1771 & 89.09 (0.23) & 2388.24 \\
\hline
\multirow{3}{*}{\textbf{30}} & 1 & 31  & 58.51 (0.45) & 30.36  \\
                             & 2 & 496 & 88.31 (0.21) & 1055.44 \\
                             & 3 & 5456 & 89.90 (0.24) & 7067.53 \\
\hline
\end{tabular}
\caption{Results for Letters dataset}
\label{tab-Letters}
\end{table}

Table~\ref{tab-Letters} presents the classification results for the Letters dataset, analyzing the effect of different numbers of principal components and neurons on accuracy and computation time. The results indicate a significant improvement in accuracy as the number of components increases, particularly when moving from 10 to 20 components, whereas the gain from 20 to 30 components is relatively small. Increasing the number of neurons also enhances performance, with the most substantial improvement occurring when transitioning from 1 to 2 neurons, while the gain from 2 to 3 neurons is less pronounced. However, this improvement comes at the cost of increased training time, which grows rapidly with the number of weights, particularly for three-neuron configurations. The highest accuracy (89.90\%) is achieved using 30 components and 3 neurons, but this setup requires a training time of 7067.53s. A slightly lower accuracy (89.09\%) is obtained with 20 components and 3 neurons in significantly less time (2388.24s), making it a more efficient choice. These results suggest that using 20 components and 3 neurons provides the best trade-off between accuracy and computational cost, as further increasing the number of components leads to only marginal performance improvements while substantially increasing training time.

\subsection{EMNIST MNIST Dataset}

The EMNIST MNIST dataset is a subset of the EMNIST dataset designed to closely resemble the original MNIST dataset while being derived from the larger NIST Special Database 19, maintaining the same structure and characteristics as MNIST. This subset serves as a benchmark for digit classification, allowing direct comparisons between classical and quantum models on a well-established dataset. Unlike other subsets of EMNIST, which include letters or merged character classes, the EMNIST MNIST dataset focuses exclusively on digit recognition, making it a valuable reference point for evaluating the performance of quantum neural networks on numerical classification tasks.

\begin{table}[H]
\centering
\begin{tabular}{|c|c|c|c|c|}
\hline
\rowcolor{gray!20}  \textbf{Components} & \textbf{Neurons} & \textbf{Weights} & \textbf{Accuracy (\%)} & \textbf{Time (s)} \\
\hline
\multirow{4}{*}{\textbf{10}} & 1 & 11  & 81.04 (0.44) & 1.81  \\
                             & 2 & 66  & 92.46 (0.27) & 6.42  \\
                             & 3 & 286 & 94.58 (0.22) & 43.79 \\
\hline
\multirow{4}{*}{\textbf{20}} & 1 & 21  & 88.46 (0.30) & 2.38  \\
                             & 2 & 231 & 97.24 (0.23) & 17.47 \\
                             & 3 & 1771 & 97.71 (0.18) & 216.79\\
\hline
\multirow{4}{*}{\textbf{30}} & 1 & 31  & 90.06 (0.25) & 2.93  \\
                             & 2 & 496 & 97.80 (0.20) & 32.51 \\
                             & 3 & 5456 & 98.12 (0.19) & 303.89\\
\hline
\multirow{4}{*}{\textbf{40}} & 1 & 41  & 91.24 (0.32) & 3.46  \\
                             & 2 & 861 & 97.79 (0.18) & 53.96 \\
                             & 3 & 12341 & 98.20 (0.17) & 628.53 \\
\hline
\end{tabular}
\caption{Results for the EMNIST MNIST dataset}
\label{tab-emnist-mnist}
\end{table}

Table~\ref{tab-emnist-mnist} presents the classification results for the EMNIST MNIST dataset, evaluating different numbers of principal components and neurons in terms of accuracy and computation time. The results show a consistent increase in accuracy as the number of components grows, with a noticeable improvement from 10 to 20 components, while the gains from 20 to 40 components are relatively small. Similarly, increasing the number of neurons enhances classification performance, with the most significant improvement occurring when moving from 1 to 2 neurons, whereas the gain from 2 to 3 neurons is more modest. The computational time also increases with model complexity, but compared to other datasets, the training time remains relatively low. The highest accuracy (98.20\%) is achieved with 40 components and 3 neurons, but a similar accuracy (98.12\%) is obtained with 30 components and 3 neurons in about half the computation time (303.89s vs. 628.53s). Given these results, the configuration with 30 components and 3 neurons provides an optimal balance between accuracy and computational efficiency, as further increasing the number of components only marginally improves performance while significantly increasing training time.

\subsection{Comparison}

In this subsection, we compare the performance of our single-qudit quantum neural network with the OPIUM classifier on the EMNIST dataset~\cite{Cohen2017}. The results are summarized in Table~\ref{table-results-summary}. The values in this table were taken from the previous tables, except for the \textit{By\_Class} and \textit{By\_Merge} subsets, which were obtained using 30 principal components and 2 neurons. Additionally, we include the accuracy obtained using Eq.~\eqref{eq-theta-}, which defines an alternative activation function based on the standard Taylor expansion, as introduced in Ref.~\cite{Souza2024}. In this equation, we made a minor modification by replacing $\tanh$ with the sigmoid function and $\cos$ with $\sin$. The results obtained using Eq.~\eqref{eq-theta-} were generated with 30 principal components and 3 neurons, the same configuration used for the results based on Eq.~\eqref{eq-ansatz}.

This analysis aims to highlight the advantages of our approach in terms of classification accuracy and efficiency. Since our model uses the high-dimensional state space of qudits for multiclass classification, it is crucial to evaluate its effectiveness against a well-established classical classifier. The results demonstrate that the qudit-based neural network, when using the activation function based on the multivariable Taylor expansion, consistently outperforms OPIUM, particularly in the classification of handwritten digits and letters. It achieves higher accuracy while maintaining a compact model structure.

\begin{table}[H]
\centering
\begin{tabular}{lllrlr}
\toprule
Dataset& OPIUM      & Standard      & Time & Multivariable & Time  \\ 
Name & Classifier & Taylor series & (sec.) & Taylor series & (sec.) \\ 
& Accuracy (\%) & Accuracy (\%) &  & Accuracy (\%) &  \\ 
\midrule
By\_Class     & $69.71 \pm 1.47$ & $67.26\pm 0.16$ & $2319.50$ & $\mathbf{82.05\pm 0.15}$ & $17091.98$ \\
By\_Merge     & $72.57 \pm 1.18$ & $69.51 \pm 0.13$ & $1320.36$ & $\mathbf{85.65 \pm 0.10}$ & $10020.33$ \\
Balanced     & $78.02 \pm 0.92$ & $63.34 \pm 0.35$ & $282.10$ & $\mathbf{81.50 \pm 0.34}$ & $21021.78$ \\
Letters      & $85.15 \pm 0.12$ & $69.38 \pm 0.28$ & $114.80$ & $\mathbf{89.90 \pm 0.24}$ & $7067.53$ \\
{\small EMNIST MNIST} & $96.22 \pm 0.14$ & $93.62 \pm 0.21$ & $10.75$ & $\mathbf{98.20 \pm 0.17}$ & $628.53$ \\
Digits       & $95.90 \pm 0.40$ & $92.84 \pm 0.16$ & $48.22$ & $\mathbf{98.88 \pm 0.07}$ & $1449.55$ \\
\bottomrule
\end{tabular}
\caption{Comparison of the accuracy of the OPIUM classifier, the standard Taylor series expansion (based on Eq.~\eqref{eq-theta-}), and the multivariable Taylor series expansion (based on Eq.~\eqref{eq-ansatz}) on the EMNIST dataset, including computation times.}
\label{table-results-summary}
\end{table}

The comparison between the Standard-Taylor-Series and Multivariable-Taylor-Series methods reveals a significant performance difference across all dataset splits. The Multivariable-Taylor-Series method consistently achieves higher accuracy than the Standard-Taylor-Series method, with improvements ranging from approximately 5\% for simpler datasets like EMNIST MNIST to over 20\% for more complex datasets like Balanced and Letters. This suggests that the Multivariable-Taylor-Series approach provides a more expressive representation, capturing nonlinear dependencies in the data more effectively.

However, this increase in accuracy comes at the cost of significantly higher computation time. The training time for the Multivariable-Taylor-Series method is considerably longer, particularly in the Balanced and Letters datasets, where it requires more than 20,000 seconds and 7,000 seconds, respectively. In contrast, the Standard-Taylor-Series method achieves lower accuracy but runs much faster, making it a more efficient alternative when computational resources are limited. These results highlight the trade-off between accuracy and computational cost when selecting an activation function for single-qudit quantum neural networks. While the Multivariable-Taylor-Series method provides superior classification performance, the Standard-Taylor-Series method remains a viable option in scenarios where computational efficiency is a priority.

\section{Classification Performance on MNIST Dataset}\label{sec-orig-mnist}

The original MNIST dataset is one of the most widely used benchmarks for handwritten digit classification in machine learning. It consists of 70,000 grayscale images of digits from 0 to 9, each represented as a 28$\times$28 pixel array. The dataset is divided into a training set of 60,000 images and a test set of 10,000 images, providing a standardized evaluation framework for different classification models. Like EMNIST, MNIST was derived from the NIST Special Database 19 with the goal of facilitating the development and evaluation of machine learning algorithms. Its simplicity, balanced class distribution, and relatively small size make it an ideal testbed for both classical and quantum-based neural networks. In this section, we apply our qudit-based quantum neural network to the original MNIST dataset and analyze its classification performance.

\begin{table}[H]
\centering
\begin{tabular}{|c|c|c|c|c|}
\hline
\rowcolor{gray!20}  \textbf{Components} & \textbf{Neurons} & \textbf{Weights} & \textbf{Accuracy (\%)} & \textbf{Time (s)} \\
\hline
\multirow{3}{*}{\textbf{10}} & 1 & 11  & 77.28 (0.51) & 1.79  \\
                             & 2 & 66 & 90.36 (0.22) & 6.61  \\
                             & 3 & 286 & 92.69 (0.16) & 33.41 \\
\hline
\multirow{3}{*}{\textbf{20}} & 1 & 21 & 85.50 (0.39) & 2.42  \\
                             & 2 & 231 & 96.20 (0.18) & 17.43 \\
                             & 3 & 1771 & 97.10 (0.13) & 118.29 \\
\hline
\multirow{3}{*}{\textbf{30}} & 1 & 31 & 87.53 (0.38) & 3.00  \\
                             & 2 & 496 & 97.02 (0.14) & 33.83 \\
                             & 3 & 5456 & 97.58 (0.18) & 312.93 \\
\hline
\end{tabular}
\caption{Results for classification performance on the original MNIST dataset. The table reports accuracy, number of neurons, number of weights, and computational time for different numbers of principal components.}
\label{tab-multivariate}
\end{table}

The results presented in Table~\ref{tab-multivariate} indicate that the classification performance of the single-qudit quantum neural network on the MNIST dataset improves as the number of principal components and neurons increases. Notably, a significant accuracy gain is observed when transitioning from 10 to 20 principal components, with the best-performing model achieving 97.10\% accuracy with 20 components and three neurons. Increasing the number of principal components to 30 yields a marginal improvement, reaching 97.58\% accuracy, but at the cost of substantially higher computational time. The results also suggest that a two-neuron configuration with 20 or 30 principal components provides a favorable trade-off between accuracy and efficiency, as it achieves near-optimal performance with significantly lower training time compared to the three-neuron setup. These findings align with previous experiments on EMNIST, reinforcing the effectiveness of leveraging qudits for multiclass classification while balancing computational resources and model complexity.

Several state-of-the-art CNN architectures have achieved near-perfect classification accuracy on the MNIST dataset, with some models reaching 99.84\% or higher~\cite{Mazzia2021,BYERLY2021,HIRATA2023}. These high-performing models employ sophisticated preprocessing techniques to enhance digit representation before training. Common preprocessing steps include normalization, deskewing, and data augmentation, which improve generalization by standardizing input distributions and introducing additional training variations. These techniques reduce distortions and enhance feature extraction, ultimately improving classification accuracy. The role of image preprocessing in CNN-based models highlights its importance in optimizing learning performance, a factor that may also impact quantum-based approaches to digit classification. In contrast, all experiments in this work were conducted on the EMNIST and original MNIST datasets without any preprocessing beyond PCA, allowing us to evaluate the raw performance of the single-qudit quantum neural network.

\section{Conclusion and Future Work}\label{sec-conc}

In this work, we introduced a single-qudit quantum neural network designed for multiclass classification tasks. By using the high-dimensional state space of qudits, our approach efficiently encodes and processes data while reducing circuit complexity compared to traditional qubit-based models. The proposed architecture employs a parameterized unitary transformation to model quantum neurons and a measurement-based classification strategy that directly maps outcomes to class labels. Additionally, the design of an activation function based on a truncated multivariable Taylor expansion enables more expressive transformations of input data, leading to improved classification performance.

We evaluated the model on the EMNIST and MNIST datasets, demonstrating that it achieves competitive classification accuracy while maintaining a compact quantum circuit structure. Our results suggest that using 30 to 40 principal components provides an optimal balance between accuracy and computational efficiency.

Despite these promising results, practical implementation faces challenges due to the current limitations of quantum hardware. While qudit-based quantum processors, such as those based on trapped ions and photonic systems, are emerging, most available platforms still rely on qubits. Future work should aim to bridge this gap by refining qubit-based approximations of qudit neurons, optimizing gate decompositions, and developing hybrid quantum-classical strategies for improved trainability and scalability. Additionally, further validation of the proposed model on real quantum hardware will be essential for assessing its practical feasibility.

As quantum computing technology advances, the integration of qudit-based architectures into quantum machine learning frameworks has the potential to improve efficiency and scalability in classification tasks. The development of noise-resilient qudit circuits and improved training techniques will be crucial to realizing the full advantages of qudit-based quantum neural networks in practical applications.

Future work should also explore the application of single-qudit quantum neural networks to regression tasks, extending their utility beyond classification. The proposed model’s ability to process high-dimensional data within a compact quantum circuit suggests potential advantages in continuous-valued predictions. Investigating how qudit-based QNNs approximate smooth functions and evaluating their expressive power in regression problems could provide deeper insights into their computational capabilities. Additionally, adapting the hybrid training strategy to optimize loss functions suited for regression, such as mean squared error (MSE), could further enhance the practical applicability of this approach. These extensions would contribute to a broader understanding of quantum neural networks in machine learning, particularly in the context of quantum data representation and function approximation.

\section*{Acknowledgements}

The work of L.~C.~Souza was supported by CNPq grant number 302519/2024-6.
The work of R.~Portugal was supported by FAPERJ grant number CNE E-26/200.954/2022, and CNPq grant numbers 304645/2023-0 and 409552/2022-4.

\section*{Declaration of competing interest}

The authors declare that they have no known competing financial interests or personal relationships that could have appeared to
influence the work reported in this paper.

\section*{Data availability}

Publicly available datasets were analyzed in this study. The datasets can be found here: \begin{itemize}
    \item https://www.nist.gov/itl/products-and-services/emnist-dataset;
    \item https://deepai.org/dataset/mnist.
\end{itemize}

\appendix

\section{}\label{app:A}
In this Appendix, we prove Eq.~\eqref{eq-psi}. The $d$-dimensional orthogonal matrix $U$ is defined as
\[ U = (A-I)(A+I)^{-1}, \]
where $A$ is a skew-symmetric matrix with non-zero entries given by
\[
A_{1\ell}=-A_{\ell 1}=\frac{s_\ell \, c_{\ell-1}}{s_1-1},
\]
where $c_\ell =\cos\theta_{d-\ell}$ and \vspace*{-10pt}
\begin{align*}
s_\ell = \prod_{k=1}^{d-\ell} \sin \theta_k,
\end{align*}
for $\ell$ ranging from 2 to $d$, with the convention that $s_d = 1$. Our goal is to show that
\begin{equation}\label{eq-Uket0}
U\ket{0} = \sum_{\ell=1}^{d} s_{\ell} c_{\ell-1}\ket{\ell-1},
\end{equation}
where we define $c_0 = 1$.

\begin{lemma}\label{lemma1}
Using the notation $c_{\ell,d}$ and $s_{\ell,d}$ to explicitly indicate the dimension of the Hilbert space, we have
\[
\sum_{\ell=1}^{d} s_{\ell,d}^2 c_{\ell-1,d}^2 \,=\,1.
\]
\end{lemma}
\begin{proof}
The proof proceeds by induction on $d$. For $d=2$, we have
\[
\sum_{\ell=1}^{2} s_{\ell,2}^2 c_{\ell-1,2}^2 \,=\,\sin^2 \theta_1+\cos^2\theta_1=1.
\]
For arbitrary $d > 2$, observe that
\[
\sum_{\ell=1}^{d} s_{\ell,d}^2 c_{\ell-1,d}^2 =(s_{1,d}^2+s_{2,d}^2 c_{1,d}^2-s_{2,d}^2)+\left(s_{2,d}^2+\sum_{\ell=3}^{d} s_{\ell,d}^2 c_{\ell-1,d}^2\right).
\]
Using the definitions of $c_{\ell,d}$ and $s_{\ell,d}$, it can be verified that the first term on the right-hand side is zero. After relabeling the dummy index, we find
\[
\sum_{\ell=1}^{d} s_{\ell,d}^2 c_{\ell-1,d}^2 =s_{2,d}^2+\sum_{\ell=2}^{d-1} s_{\ell+1,d}^2 c_{\ell,d}^2.
\]
Since $s_{\ell+1,d} = s_{\ell,d-1}$ and $c_{\ell,d} = c_{\ell-1,d-1}$, this becomes 
\[
\sum_{\ell=1}^{d} s_{\ell,d}^2 c_{\ell-1,d}^2 =\sum_{\ell=1}^{d-1} s_{\ell,d-1}^2 c_{\ell-1,d-1}^2.
\]
By the inductive hypothesis, the right-hand side is equal to 1.
\end{proof}

To prove Eq.~\eqref{eq-Uket0}, it suffices to compute the first column of $U$. This requires calculating only the first column of $B = (A+I)^{-1}$. We claim that the first column of $B$ is given by
\begin{align*}
B_{11} &= \frac{1-s_1}{2}, \\
B_{\ell 1} &= -\frac{s_{\ell} , c_{\ell-1}}{2},
\end{align*}
for $\ell$ ranging from 2 to $d$. To prove this, we need to show that
\[
\sum_{\ell=1}^d (A+I)_{k \ell} B_{\ell 1} = \delta_{k 1}
\]
for $1\le k \le d$. For $k=1$, we have
\[
\sum_{\ell=1}^d (A+I)_{1 \ell} B_{\ell 1} = \frac{1-s_1}{2}-\frac{1}{2(s_1-1)}\sum_{\ell=2}^d s_\ell^2 c_{\ell -1}^2.
\]
Using Lemma~\ref{lemma1}, it follows that the right-hand side is 1. For $k > 1$, we find
\[
\sum_{\ell=1}^d (A+I)_{k \ell} B_{\ell 1} = A_{k1}B_{11}+B_{k1}.
\]
Using the definitions of $A$ and $B$, we confirm that the right-hand side is 0, completing the proof of the claim.

To compute the first column of $U$, we multiply $(A-I)$ by the first column of $B$ as follows:
\[
U_{k1}=\sum_{\ell=1}^d (A-I)_{k \ell} B_{\ell 1}.
\]
Using calculations similar to those above, we find $U_{11} = s_1$ and $U_{\ell 1} = s_{\ell} c_{\ell-1}$ for $\ell > 1$, completing the proof of Eq.~\eqref{eq-Uket0}.

\section{}\label{app:B}

In this Appendix, we show how to implement a qubit-based circuit that generates the state described in Eq.~\eqref{eq-psi}. The circuit employs $d-1$ qubits and applies the standard state-preparation method. Fig.~\ref{fig:qubit-circ} illustrates an example of a circuit implementing the 5-dimensional $U(\theta_1, \theta_2, \theta_3, \theta_4)$, which can be easily generalized for the $d$-dimensional case. The circuit uses multi-controlled $R_y$ gates, where 
\[R_y(\theta) = \left[\begin{array}{cc}
 \cos\frac{\theta}{2} &  -\sin\frac{\theta}{2} \vspace{2pt}\\
 \sin\frac{\theta}{2}   &  \,\,\,\,\cos\frac{\theta}{2}
\end{array}\right].
\]
An empty control indicates that the target gate is active when the control qubit is in the $\ket{0}$ state.

\begin{figure}[H]
\centering\includegraphics[scale=0.19]{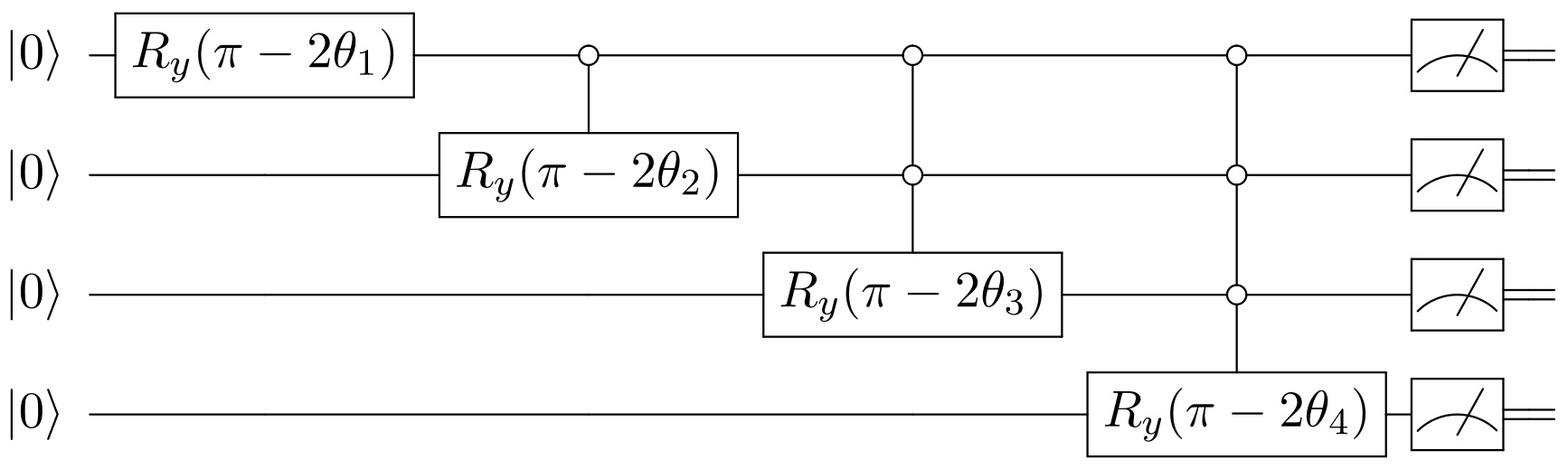}
\caption{Example of a qubit-based circuit implementing the 5-dimensional unitary $U(\theta_1, \theta_2, \theta_3, \theta_4)$ described in Sec.~\ref{sec-neuron} using multi-controlled $R_y$ gates.}
\label{fig:qubit-circ}
\end{figure}

In this setup, as shown in Fig.~\ref{fig:qubit-circ}, the first class is associated with the measurement output 0000, the second with 0001, and so on, up to the fifth class with 1000. If the quantum computer is error-free, all other outputs will have zero probability. We use more qubits than strictly necessary to implement a 5-dimensional unitary operator because the circuit is simpler, and the redundancy also serves as an error detection mechanism. Any measurement outcome other than the expected ones indicates an error, requiring the classification process to be restarted. However, a bit flip error affecting the output 00000 would go undetected, as it falls within the set of expected results.

Let us now generalize to the $d$-dimensional case. For $d=2$, the output of the circuit before measurement is $\sin \theta_1\ket{0} + \cos \theta_1\ket{1}$, as expected. To use induction, assume that for $d-1$ qubits, the output is
\begin{equation}\label{eq-psi-power2}
s_{1,d}\ket{0}+\sum_{j=1}^{d-1} s_{j+1,d}c_{j,d}\ket{2^{j-1}},
\end{equation}
where we use decimal notation inside the kets. This is a straightforward extension of Eq.~\eqref{eq-psi}. By adding an extra qubit and applying a multi-controlled $R_y(\pi-2\theta_{d})$ gate with $d-1$ empty controls, only the first term of the equation above is modified, resulting in
\[
s_{1,d}\sin \theta_d\ket{0}\otimes \ket{0}+s_{1,d}\cos \theta_d\ket{0}\otimes \ket{1}+\sum_{j=1}^{d-1} s_{j+1,d}c_{j,d}\ket{2^{j-1}}\otimes \ket{0}.
\]
After relabeling the dummy index and using that $s_{j+2,d+1} = s_{j+1,d}$ and $c_{j+1,d+1} = c_{j,d}$, we obtain
\[
s_{1,d+1}\ket{0}+\sum_{j=1}^{d} s_{j+1,d+1}c_{j,d+1}\ket{2^{j-1}}.
\]
Since this matches the form of Eq.~\eqref{eq-psi-power2} for $d+1$, the proof is complete.

By applying the decomposition techniques described in Refs.~\cite{Saeedi2013,Silva2022} for multi-controlled $R_y$ gates, the circuit implementing $U(\theta_1, \ldots, \theta_{d-1})$ with qubits achieves a depth that scales as $O(d^2)$.

\newpage


\end{document}